\documentclass{amsart}
\usepackage{tikz-cd}
\usepackage{latexsym,amssymb}
\usepackage[english]{babel}
\usepackage{bm}
\usepackage{amsfonts, amsthm}
\usepackage{amsmath}
\usepackage{mathrsfs}
\usepackage{multirow}
\usepackage{bm}
\usepackage{pgf,tikz}

\usepackage[ruled, vlined, linesnumbered, noresetcount]{algorithm2e}

\allowdisplaybreaks
\newtheorem{thm}{Theorem}[section]
\newtheorem{lem}[thm]{Lemma}
\newtheorem{cor}[thm]{Corollary}

\newtheorem{rem}[thm]{Remark}

\theoremstyle{definition} 
\newtheorem{defn}[thm]{Definition}

\DeclareMathOperator\supp{supp}
\newcommand{\remove}[1]{}

\newcommand{\bc}{{\bf  c}}

\begin{document}
\title[An Efficient  Algorithm for Group Testing with Runlength Constraints]{An Efficient  Algorithm for Group Testing with Runlength Constraints}

\author[M.Dalai]{Marco Dalai}
\address{DII, Universit\`a degli Studi di Brescia, Via
Branze 38, I-25123 Brescia, Italy}
\email{marco.dalai@unibs.it}
\author[S.DellaFiore]{Stefano Della Fiore}
\author[A.Rescigno]{Adele A. Rescigno}
\author[U.Vaccaro]{Ugo Vaccaro}

\address{DI, Universit\`a degli Studi di Salerno, Via Giovanni Paolo II 132, 84084 Fisciano, Italy}
\email{sdellafiore@unisa.it, arescigno@unisa.it, uvaccaro@unisa.it}


\begin{abstract}
In this paper, we provide an efficient  algorithm to construct almost 
optimal $(k,n,d)$-superimposed codes with runlength constraints. A $(k,n,d)$-superimposed code of length $t$ is a $t \times n$ binary matrix such that any two 1's  in each column are separated by a run of at least $d$ 0's, and such that for any column $\mathbf{c}$ and any other $k-1$ columns, there exists a row where $\mathbf{c}$ has $1$ and all the remaining $k-1$ columns have $0$. These combinatorial structures  were introduced by Agarwal {\em et al.} \cite{Olgica}, in the context of Non-Adaptive Group Testing algorithms with runlength constraints.

By using Moser and Tardos' constructive version of the 
Lov\'asz Local Lemma, we provide an efficient
randomized Las Vegas algorithm of complexity $\Theta(t n^2)$ for the construction of $(k,n,d)$-superimposed codes of length $t=O(dk\log n +k^2\log n)$. 
We also show that the length of our codes is shorter, for $n$ sufficiently large,
than that of the codes whose existence was proved in \cite{Olgica}.
\end{abstract}

\maketitle

\section{Introduction}\label{sec:Intro}
In this paper, we devise  efficient construction  algorithms for 
$(k, n,d)$-\emph{superimposed} codes recently introduced by Agarwal \emph{et al.} in \cite{Olgica} and defined as follows:

\begin{defn}[\cite{Olgica}]\label{def:dsuper}
Let $k$, $n$, $d$ be positive integers,  $k \leq n$.
A $(k, n,d)$-\emph{superimposed} code is a $t \times n$ binary matrix $M$ such that
\begin{itemize}
    \item[1)]any two 1's  in each column of $M$  are  separated by a run of  at least $d$ 0's,  
    \item[2)] for any $k$-tuple of the columns of $M$ and for any 
    column $\mathbf{c}$ of the \emph{given} $k$-tuple, it holds that
    there exists a row ${i \in \{1,\ldots , t\}}$ such that $\mathbf{c}$ has symbol $1$ in row $i$ and all the remaining $k-1$ columns of the $k$-tuple have symbol $0$ in row $i$.
\end{itemize}The number of rows $t$ of $M$ is called the length of the $(k,n,d)$-superimposed code.
\end{defn}

$(k, n,d)$-{superimposed} codes  were introduced within  the context of 
Non-Adaptive Group Testing algorithms for topological DNA-based data storage,  
and represent 
one of the main instruments to derive the strong results obtained 
therein \cite{Olgica} (see also \cite{T}). The parameter of $(k, n,d)$-{superimposed} codes 
that one wants  to optimize (i.e., minimize) is the length $t$ of the code.
Indeed,  this is the parameter that 
mostly affects the DNA-based data storage  algorithms considered in  \cite{Olgica}.
Using the probabilistic method the authors of \cite{Olgica}
proved that $(k, n,d)$-{superimposed} codes
of length 
$t=O(dk\log n +k^2\log n)$ exists and they provided a randomized Montecarlo algorithm, in the sense that it gives, with high probability,
a $(k, n, d)$-superimposed code whose length is upper bounded by this quantity. They also proved that any 
$(k, n,d)$-{superimposed} code must have length 
\begin{equation}\label{eq:lwo}
t\geq \min\left (n, \Omega\left(\frac{d k}{\log (d k)}\log n + \frac{k^2}{\log k} \log n\right) \right).
\end{equation}

A preliminary study of the questions treated in the present paper was done in \cite{DDV} where we improved some of the existential bounds  of \cite{Olgica} by showing the existence of 
$(k, n,d)$-{superimposed} codes having  shorter lengths than the codes~of~\cite{Olgica}. In \cite{Olgica} and \cite{DDV}, the authors left open the problem of devising 
an efficient polynomial time algorithm to construct $(k, n,d)$-{superimposed} codes
of length $t=O(dk\log n +k^2\log n)$. More precisely, the results 
of \cite{Olgica} and \cite{DDV} only imply the existence of
$\Theta(n^k)$-time algorithms for constructing $(k,n,d)$-{superimposed} codes
of length $t=O(dk\log n +k^2\log n)$. We note that such algorithms achieve a time-complexity of $\Theta(n^k)$ since in order to see if a randomized constructed matrix is a $(k,n,d)$-superimposed code they need to check conditions that involved $k$-tuple of columns. It is clear that already for moderate values of
$k$, those algorithms are impractical. In view of the relevance  of the
application scenario considered in \cite{Olgica}, it is quite important
to have an  efficient algorithm for  constructing $(k, n,d)$-{superimposed} codes
of length $t=O(dk\log n +k^2\log n)$.
The purpose of this paper is to provide a randomized Las Vegas $\Theta(k(k+d)n^2\ln n)$-time algorithm to construct such codes that is polynomial both in $n$ and $k$. We remark that
our algorithm produces almost optimal codes (in the asymptotic sense)
because of the lower bound (\ref{eq:lwo}).

In the same spirit of this work, in \cite{Cheng2} and \cite{Yeh}, the authors provided using probabilistic methods, such as the Lov\'asz Local Lemma, new bounds on the length of $(k+1,n,0)$-superimposed codes (also known as $k$-disjunct matrices) with fixed-weight columns. Their approaches can be adapted to derive $\Theta(n^k)$-time algorithms and bounds on the length of $(k,n,d)$-superimposed codes. This further motivates our work in studying algorithms to construct $(k,n,d)$-superimposed codes that are polynomial both in $n$ and $k$.

Before going into the technical details, we would like to recall that 
$(k,n,0)$-{superimposed} codes correspond to the classical superimposed codes 
(a.k.a. cover free families)
introduced in \cite{KS,Erdos}, and extensively studied since then.
We refer to the excellent survey papers \cite{Deppe,IM} for a broad 
discussion of the relevant literature, and to the monographs \cite{DuHwang,JSA}
for an account of the  applications of superimposed codes to group testing, 
multi-access communication, data security, data compression, and several 
other different areas. It is  likely that 
also $(k, n,d)$-{superimposed} codes
will find applications    outside the original  scenario considered in 
\cite{Olgica}.

\section{Preliminaries}

Throughout the paper, the logarithms without subscripts are in base two, and we denote with $\ln (\cdot)$ the natural logarithm.  We denote by $[a,  b]$ the set ${\{a, a+1, \ldots, b\}}$.
 Given integers $w$ and $d$, a binary $(w,d)$-vector $\mathbf{x}$ is 
 a  vector of Hamming weight $w$ (that is, the number of $1$'s in 
 $\mathbf{x}$ is equal to $w$), such that any 
 two $1$'s in $\mathbf{x}$ are separated by a run of at least d $0$’s.

\remove{\begin{defn}[\cite{Olgica}]
Let $k$, $n$, $d$ be positive integers,  $k \leq n$.
A $(k, n,d)$-\emph{super\-\-imposed} code is a $t \times n$ binary matrix $M$ such that any two 1's  in each column of $M$  are  separated by a run of  at least $d$ 0's,  and for any $k$-tuple of the columns of $M$ we have that for any column $\mathbf{c}$ of the given $k$-tuple, there exists a row ${i \in [1,t]}$ such that $\mathbf{c}$ has symbol $1$ in row $i$ and all the remaining $k-1$ columns of the $k$-tuple are equal to $0$.  The number of rows $t$ of $M$ is called the length of the $(k,n,d)$-superimposed code.
\end{defn}
}

We recall, for positive integers $c \leq b \leq a$,
 the following well-known properties  of binomial coefficients:
\begin{equation}\label{eq:ULBinom}
    \left(\frac{a}{b}\right)^b \leq \binom{a}{b} \leq \frac{a^b}{b!} \leq \left(\frac{ea}{b}\right)^b\,,
\end{equation}
\begin{equation}\label{eq:idBinom}
    \binom{a}{b} \binom{b}{c} = \binom{a}{c} \binom{a-c}{b-c}\,.
\end{equation}
We shall also need  the following technical lemma from  \cite{Vaccaro1}.  

\remove{\begin{lem}\label{lem:firstBoundBinomial}
Let $a, b, c$ be positive integers such that $c \leq a \leq b$. We have that
$$
	\frac{a}{b} \cdot \frac{a-c}{b-c} \leq \left(\frac{a - \frac{c}{2}}{b-\frac{c}{2}}\right)^2\,.
$$
\end{lem}
\begin{proof}
Clearly $a (a-c) c^2 \leq b (b-c) c^2$. Then adding the quantity $4ab (a-c)(b-c)$ to both members implies that $a(a-c) (2b-c)^2 \leq b(b-c) (2a-c)^2$. Therefore, Lemma~\ref{lem:firstBoundBinomial} follows.
\end{proof}
}
\begin{lem}[\cite{Vaccaro1}]\label{cor:boundBinomial}
Let $a, b, c$ be positive integers such that $ c \leq a \leq b$.  We have that
\begin{equation}\label{eq:boundB}
	\frac{\binom{a}{c}}{\binom{b}{c}} \leq \left(\frac{a - \frac{c-1}{2}}{b - \frac{c-1}{2}}\right)^{c}\,.
\end{equation}
\end{lem}

\remove{
\begin{proof}
Expanding the LHS of \eqref{eq:boundB} we get
\begin{equation}\label{eq:expBin}
	\frac{\binom{a}{c}}{\binom{b}{c}} = \frac{a}{b} \cdot \frac{a-1}{b-1} \cdots \frac{a-c+1}{b-c+1}\,.
\end{equation}
Let us group the terms in \eqref{eq:expBin} into pairs as follows
\begin{equation}\label{eq:termsBin}
	\frac{a-i}{b-i} \cdot \frac{a- (c-i-1)}{b-(c-i-1)} \text{ for } i=0,\ldots, \left\lceil \frac{c-1}{2} \right\rceil -1 \,.
\end{equation}
If $c$ is odd then we leave alone the term $(a-\frac{c-1}{2})/(b-\frac{c-1}{2})$.  By Lemma \ref{lem:firstBoundBinomial},  each term in \eqref{eq:termsBin} can be upper bounded by
$$
	\frac{a-i}{b-i} \cdot \frac{a- (c-i-1)}{b-(c-i-1)} \leq \left(\frac{a - \frac{c-1}{2}}{b-\frac{c-1}{2}}\right)^2\,.
$$
Hence Corollary \ref{cor:boundBinomial} follows.
\end{proof}
}

Finally, we  recall  here the celebrated algorithmic version of the 
Lov\'asz Local Lemma for the symmetric case, due to Moser and Tardos \cite{MT}.
It represents one of the main tools to derive the results of this paper. 
We first recall the setting for the Lov\'asz Local Lemma in 
the random-variable
scenario. The relevant probability space $\Omega$ is defined by 
$n$ mutually independent random variables $X_1, \ldots , X_n$, taking values in a finite
set $\mathcal{X}$. One is interested in a set of
events $\mathcal{E}$ in the probability space $\Omega$, 
(generally called \lq\lq bad events", that is, events one wants to
avoid), where each event $E_i\in \mathcal{E}$
 only depends on 
$\{X_j : j \in S_i\}$ for some subset $S_i\subseteq [1,n]$, for 
$i=1, \ldots, |\mathcal{E}|$.
Note that two events $E_i,E_j$ are independent  if $S_i\cap S_j=\emptyset.$ 
A \textit{configuration} in the context of the Lovász Local Lemma is a specific assignment of values to the set of random variables involved in defining the events.
Sampling a random variable $X_i$ means generating a
value $x\in \mathcal{X}$ in such a way that the probability of generating 
$x$ is in accordance with  the probability distribution of 
the random variable $X_i$.

As said before, in the applications of the  Lov\'asz Local Lemma 
the events $E_i$'s are bad-events that one wants to avoid, that is, one seeks
a configuration such that all the events  $E_i$'s do not hold.
In the seminal paper \cite{MT}  Moser and Tardos 
introduced a simple randomized algorithm 
that produces such a configuration, under the same hypothesis 
of the classical Lovász Local Lemma. The algorithm is the following:

\begin{algorithm}[!ht]

Sample the random variables $X_1, \ldots , X_n$ from their distributions in $\Omega$

\While{\rm some event is true on $X_1, \ldots , X_n$}
{Arbitrarily select some true event $E_i$\\
For each $j\in S_i$, sample $X_j$ from its distribution in $\Omega$}

\caption{The MT algorithm }
\label{algo0}
\end{algorithm}

 Moser and Tardos \cite{MT} proved the following important result
 (see also  \cite{knuth}, p. 266).

\begin{lem}[\cite{MT}] \label{lem:LLL}
Let $\mathcal{P}$ be a finite set of mutually independent random variables in a probability space and let $\mathcal{E}= \{E_{1},E_{2},\ldots, E_{m}\}$ be a set of $m$ events where each $E_i$ is determined by a subset $S_i$ of the random variables and, for each $i$, $S_j\cap S_i\neq\emptyset$ for at most $D$ values of $j\neq i$. 
Suppose that $\Pr(E_{i})\leq P$ for all $1\le i\leq m$.
If $eP D\leq1$, then $\Pr(\cap_{i=1}^{m}\overline{E_{i}})>0$.
Moreover, 
\textbf{\rm \bf Algorithm \ref{algo0}} finds a configuration avoiding all events $E_i$ 
by using an average number 
of resampling of at  most $m/D$.
\end{lem}

\section{New Algorithms for $(k,n,d)$-superimposed  codes}\label{sec:superimposed}

We aim to efficiently construct $(k,n,d)$-superimposed codes with a small length. 
The difficulty faced in \cite{Olgica,DDV} was essentially due to the fact
that the constraints  a $t\times n$ binary matrix has to satisfy in order to
be a $(k,n,d)$-superimposed code involve all the $\binom{n}{k}$
$k$-tuples of
columns of $M$. Checking  whether those conditions are satisfied or not
requires time $\Theta(n^k)$. To overcome this difficulty, we use an idea 
of \cite{Erdos,KS}. That is,  we first introduce an auxiliary class of binary
matrices where the constraints involve {\em only} pairs of columns. Successively,
we show that for suitably chosen parameters such a class of matrices 
give rise to 
$(k,n,d)$-superimposed codes with small length $t$. This  implies that we need to check 
the validity of the constraints only for the  $\Theta(n^2)$ pairs of columns.
This observation and   Lemma \ref{lem:LLL}, will allow
us to provide an  efficient algorithm to construct $(k,n,d)$-superimposed codes.

It is convenient to first 
consider the following class of $(k,n,d)$-superimposed codes.
\begin{defn}
A $(k,n,d,w)$-\emph{superimposed} code is a $(k,n,d)$-superimposed code with the additional constraint that each column has Hamming weight $w$, that is, 
each column of the code  is a binary $(w,d)$-vector.
\end{defn}

We now introduce the auxiliary class of matrices 
mentioned above.
\begin{defn} \label{matrix}
Let $n, d, w, \lambda$ be positive integers. A $t \times n$ binary matrix $M$ is a $(n,d,w,\lambda)$-matrix if the following properties hold true:
\begin{enumerate}
    \item each column of $M$ is a binary $(w, d)$-vector;
    \item any pair of columns $\mathbf{c}, \mathbf{d}$ of $M$ have at most $\lambda$ $1$'s in common, that is, there are at most $\lambda$ rows among the $t$'s  where columns $\mathbf{c}$ and $\mathbf{d}$  both 
    have~symbol~$1$.
\end{enumerate}
\end{defn}

$(n,d,w,\lambda)$-matrices are related to $(k,n,d,w)$-superimposed codes 
 by way of the following easy result.
\begin{lem}\label{LaToSup}
A binary $(n,d,w, \lambda)$-matrix   $M$ of dimension 
$t\times n$, with parameter 
$\lambda=\left\lfloor (w-1)/(k-1)\right\rfloor$, 
 is a $(k,n,d,w)$-superimposed code of length $t$.
\end{lem}
\begin{proof}
The bitwise OR of any set $C$ of $k-1$ columns of $M$ can have at most
$(k-1)\lambda=(k-1)\left\lfloor (w-1)/(k-1)\right\rfloor<w$ 
symbols equal to 1 in the same $w$ rows where an arbitrary column 
$\mathbf{c}\notin C$ has a $1$.
\end{proof}

\remove{
\begin{proof}
Let $\mathbf{c}$ be an arbitrary column of $M$ and let $A$ be the set of row indices in which column $\mathbf{c}$ has a $1$. Therefore $|A| = w$. Let $K$ be a set of arbitrary $k-1$ columns of $M$ where $\mathbf{c} \not \in K$. Let us denote by $M(A,K)$ the $w \times (k-1)$ submatrix of $M$ constructed by first selecting the $k-1$ columns in $K$ and then selecting the $w$ rows whose indices belong to $A$. Since $M$ is a $(n,d,w,\lambda)$-matrix we have that the number of $1$'s that column $\mathbf{c}$ share (in the same rows) with any column in $K$ is at most $\lambda = \left\lfloor \frac{w-1}{k-1} \right\rfloor$. Hence, the total number of $1$'s in $M(A,K)$ is at most $$(k-1) \lambda \leq (k-1) \left\lfloor \frac{w-1}{k-1} \right\rfloor \leq w-1\,.$$ Considering that the matrix $M(A,K)$ has 
$w$ rows, we have that at least one row in $M(A,K)$ contains only zero elements. 
Then the lemma follows.
\end{proof}
}

We now  show how to efficiently construct binary $(n,d,w,\lambda)$-matrices with a small number of rows. This fact, by virtue of Lemma \ref{LaToSup}, will 
give us an upper bound on the minimum length of $(k,n,d,w)$-superimposed codes.

We need the following enumerative lemma from \cite{DDV}. We include
here the 
short proof
to keep  the paper self-contained. Since $(w,d)$-vectors of 
length $t$ have necessarily $t \geq (w-1)d + w$, we use this inequality
throughout the paper.

\begin{lem}\label{lem:enumerative}
Let $V \subseteq \{0,1\}^t$ be the set of all binary $(w, d)$-vectors
of length $t$. Then
$$
	|V| = \binom{t - (w-1) d}{w}\,.
$$
\end{lem}
\begin{proof}
Let $A$ be the set of all distinct binary vectors of length $t-(w-1)d$ and weight $w$.
One can see that $|V| = |A|$ since each vector  of $V$ can be obtained from an element $a \in A$ by inserting a run of  exactly $d$ 0's
between each pair of $1$'s in $a$.  Conversely,  each element of $A$ can be obtained from an element $s \in V$ by removing exactly $d$ consecutive 0's
in 
between each pair of consecutive $1$'s in $s$. 
\end{proof}

We are ready to state one of the main results of this paper.

\begin{thm}\label{th:LLPairs}
There exists a $t \times n$ $(n,d,w,\lambda)$-matrix with
\begin{equation}\label{eq:boundLL}
t = \left \lceil  (w-1)d + \frac{\lambda}{2} + \frac{ew}{\lambda+1}\left(w-\frac{\lambda}{2} \right) (e(2n-4))^{\frac{1}{\lambda+1}}
\right \rceil\,.
\end{equation}
\end{thm}

\begin{proof}
Let $M$ be a $t \times n$ binary matrix, $t \geq (w-1)d+w$, where each column $\bm c$ is sampled uniformly at random among the set of all distinct binary $(w, d)$-vectors of length $t$. Since we are assuming that $t \geq (w-1)d+w$, by Lemma \ref{lem:enumerative} we have that
$$\Pr(\mathbf{c}) = \binom{t - (w-1) d}{w}^{-1}\,.$$
Let $i,j\in [1, n], i\neq j$ and let us consider the event $\overline{E}_{i,j}$ that there exists \emph{at most} $\lambda$ rows  
such that both the $i$-th column and the $j$-th column of $M$ have the symbol $1$ in \emph{each} of these rows.
We  evaluate the probability of the complementary ``\emph{bad}'' event ${E}_{i,j}$. Hence $E_{i,j}$ is the event that the
random $i$-th and $j$-th columns $\bc_i$ and $\bc_j$ have $1$ in at least $\lambda+1$ coordinates. We bound $\Pr(E_{i,j})$ by conditioning on the event that $\bc_i$ is equal to a $c$, where $c$ is a binary $(w,d)$-vector.

For a subset $S\subset [1,t]$ of coordinates, let $E_{i,j}^S$ be the event that in each coordinate of $S$ the $i$-th and $j$-th column have the symbol $1$. 
For a fixed column $\bc_i=c$, let $A$ 
be the set of coordinates where  $c$ has $1$'s.
Note that for $S\in\binom{A}{\lambda+1}$, i.e., for a  subset $S$ of $A$ of size $\lambda +1$, it holds that
\begin{equation}\label{eq:upperProbS}
\Pr(E_{i,j}^S|\bc_i = c) \leq \frac{\binom{t-(w-1)d-(\lambda+1)}{w-(\lambda+1)}}{\binom{t - (w-1) d}{w}}\,.
\end{equation}
We justify \eqref{eq:upperProbS}. 
Since by assumption  $c$ already contains a $1$ in each coordinate of $S\subset A$, given that $\bc_i=c$ we have  the event $E_{i,j}^S$ conditionally reduces to the event that $\bc_j$ also has $1$'s in all the $\lambda+1$ coordinates in $S$. Therefore,  we only need to upper bound the number of $(w,d)$-vectors of length $t$ with $1$'s in the $\lambda+1$ coordinates of $S$. Note that each such $t$-long vector, upon removing exactly $d$ $0$'s in between each pair of consecutive $1$'s and the coordinates in $S$, reduces to a distinct binary vector of length $t-(w-1)d-(\lambda+1)$ and weight $w - (\lambda+1)$. It follows that the number of choices for $\bc_j$ (such that it 
has $1$'s in all the $\lambda+1$ coordinates in $S$) is upper bounded by $\binom{t-(w-1)d-(\lambda+1)}{w - (\lambda+1)}$. Then, 
formula \eqref{eq:upperProbS} holds.

\remove{
Given a choice of $\bm c_i = c$, we determine an upper bound on the probability of the event $E_{i,j}^S$ given $\bm c_i = c$. Since $\bm c_i$ and $\bm c_j$ have $1$ in all  coordinates in $S$, the remaining $w-(\lambda+1)$ $1$'s of $\bm c_j$ can be put only in some of the $t-w$ coordinates of $\bm c_i$. However, $\bm c_j$ has at least $(w-1)d$ $0$'s, and at most $w-(\lambda+1)$ of them can be in the same coordinates where $\bm c_i$ has a $1$. Therefore, among the $t-w$ where the $1$'s of $\bm c_j$ could possibly put, there are at least $(w-1)d - (w - (\lambda+1))$ coordinates that must be necessarily occupied by the $0$'s of $\bm c_j$. Wrapping all together, the number of positions in which we can put the $w-(\lambda+1)$ $1$'s of $\bm c_j$ is no more than $t-w-((w-1)d - (w - (\lambda+1))) = t-(w-1)d-(\lambda+1)$.
}
Therefore, it holds that
\begin{align}
\Pr(E_{i,j}|\bc_i = c) & \leq \sum_{S\in \binom{A}{\lambda+1}}\Pr(E_{i,j}^S|\bc_i = c) \nonumber \\& = \binom{w}{\lambda + 1}\frac{\binom{t-(w-1)d-(\lambda+1)}{w-(\lambda+1)}}{\binom{t - (w-1) d}{w}}.\label{eq: upperEijS}
\end{align}
Since the right-hand side of \eqref{eq: upperEijS} does not depend on the fixed column $c$, it also holds unconditionally. Hence
\begin{equation}\label{eq:bounPrS}
\Pr(E_{i,j}) \leq \binom{w}{\lambda + 1}\frac{\binom{t-(w-1)d-(\lambda+1)}{w-(\lambda+1)}}{\binom{t - (w-1) d}{w}}.
\end{equation}
Hence, by \eqref{eq:bounPrS} we have
\begin{align}
\Pr({E}_{i,j}) &\leq \binom{w}{\lambda + 1}\binom{t-(w-1)d-(\lambda+1)}{w-(\lambda+1)} \Big/ \binom{t - (w-1) d}{w} \nonumber \\
&\stackrel{(i)}{=} \binom{w}{\lambda + 1} \binom{w}{\lambda + 1} \Big/ \binom{t - (w-1) d}{\lambda + 1} \nonumber \\
&\stackrel{(ii)}{\leq} \binom{w}{\lambda + 1} \left( \frac{w - \frac{\lambda}{2}}{t- (w-1) d - \frac{\lambda}{2}} \right)^{\lambda + 1} \nonumber \\
&\stackrel{(iii)}{\leq} \left(\frac{ew}{\lambda + 1} \right)^{\lambda+1} \left( \frac{w - \frac{\lambda}{2}}{t - (w-1) d - \frac{\lambda}{2}} \right)^{\lambda + 1} =P\,,
\label{Q-twoS}
\end{align}
where $(i)$ holds due to equality \eqref{eq:idBinom} (since 
$t \geq (w-1)d + w$), $(ii)$ is true due to Lemma \ref{cor:boundBinomial}, and finally $(iii)$ holds thanks to inequalities~\eqref{eq:ULBinom}.

The number of events ${E}_{i,j}$ is equal to $n(n-1)/2$.
Let us fix an event $E_{i,j}$. Then the number
of events $E_{i',j'}$ with  $\{i, j\} \cap \{i', j'\} \neq \emptyset$ and $\{i,j\} \neq \{i',j'\}$  is equal to $D=2n-4$. Hence, according to Lemma \ref{lem:LLL}, if we take $\mathcal{P}=\{\bc_1, \bc_2, \ldots, \bc_n\}$ to be the set of $n$ mutually independent random variables that represent the columns of the matrix $M$, $\mathcal{E} = \{E_{i,j}\}$ to be the set of events defined earlier that are associated to these random variables (where each $E_{i,j}$ is only determined by $\bc_i$ and $\bc_j$),  $P$  (as defined in \eqref{Q-twoS}) and $D=2n-4$ that satisfies $e P D\leq 1$, then the probability that {\em none} of the ``bad"  events $E_{i,j}$ occurs is strictly  positive. By solving the following inequality for $t$
$$e P D=  e(2n-4) \left (\frac{ew}{\lambda+1} \right )^{\lambda+1} \left (\frac{w-\frac{\lambda}{2}}{t-(w-1)d-\frac{\lambda}{2}} \right )^{\lambda+1} < 1\,,$$
one can see that by setting  $t$ as in \eqref{eq:boundLL} we are indeed satisfying this inequality.
We also note that the initial condition $t \geq (w-1)d+w$ is satisfied for this value of $t$, as given in \eqref{eq:boundLL}.

Hence, from Lemma \ref{lem:LLL}  one can construct a binary $(n,d,w,\lambda)$-matrix $M$ whose number of rows $t$ satisfies equality~\eqref{eq:boundLL}.
\end{proof}

Now, thanks to Lemma \ref{LaToSup} and Theorem \ref{th:LLPairs}, we can prove the 
following result.

\begin{thm}\label{thm:algo}
   
There exists a randomized algorithm to construct  a $(k,n,d,w)$-superimposed code with length
    \begin{equation}
    t \leq 1 + (w-1)d + \frac{w-1}{2(k-1)} + \\ \frac{ew(k-1)}{w-1}\left(w-\frac{w-1}{2(k-1)} +\frac{1}{2} \right) (e(2n-4))^{\frac{k-1}{w-1}}. \label{eq:upperWstrong}
\end{equation}
The algorithm requires, on average, time $O(t n^2)$ to construct the code.
\end{thm}
\begin{proof}
The upper bound (\ref{eq:upperWstrong}) on $t$  is derived by substituting the value of $\lambda = \lfloor (w-1)/(k-1) \rfloor$ from Lemma \ref{LaToSup} into equation \eqref{eq:boundLL} of Theorem \ref{th:LLPairs}, and by using the inequalities $\frac{w-1}{k-1}-1 \leq \left\lfloor \frac{w-1}{k-1} \right\rfloor \leq \frac{w-1}{k-1}$.
The time complexity $O(t n^2)$ comes from Lemma \ref{lem:LLL} by first noticing that $m/D = n(n-1)/(4n-8) \leq n/3$, for $n\geq 5$. Moreover,  {\rm \bf  Algorithm \ref{algo0}}   requires to randomly generate a matrix, checking if the $\Theta(n^2)$ events $\overline{E}_{i,j}$ are satisfied, and resampling  \emph{only on non-satisfied} events. 
The generation of each matrix-column can be done 
by first generating an integer in the interval $[0, \binom{t - (w-1) d}{w}-1]$
uniformly at random,
and encoding it with a different  
binary vector of length $t-(w-1)d$ containing  $w$ 1's. This 
one-to-one encoding can be
performed in time $O(t)$ for each column, 
by using the enumeration encoding technique by Cover \cite{Cover}.
Successively, one inserts a run of exactly $d$ $0$'s between each pair of $1$'s
so that each matrix-column is a $(w,d)$-vector. All together, this 
requires $O(t)$ operations per column.

In order to check whether an arbitrary event $\overline{E}_{i,j}$ is satisfied, 
we need to check whether the $i$-th column and the $j$-th column of the matrix have at most $\left\lfloor\frac{w-1}{k-1}\right\rfloor$ $1$'s in common; this  can be done with at most $O(t)$ operations. Successively, we resample only over non-satisfied events. Then, we need to check only the events that involve columns that have been resampled. Altogether, by Lemma \ref{lem:LLL} this procedure requires  
$O(tn^2 + n \cdot m/D \cdot t ) = O(t n^2)$ elementary operations.
\end{proof}

Now, we optimize the parameter $w$ in equation \eqref{eq:upperWstrong} to 
obtain a randomized algorithm for (weight-unconstrained) $(k,n,d)$-sumperimposed codes.

\begin{thm}\label{thm:algosuperimposed}
There exists a randomized algorithm to construct a $(k,n,d)$-super\-\-imposed code with length
    \begin{equation*}
  t \leq d(k-1) \ln(2en) +  \frac{\ln(n)}{2} +  e^2 (k-1)^2 \ln(2en) + \frac{7e^2(k-1)}{2} + d+O(1).
\end{equation*}
    The algorithm requires, on average, time $O(t n^2)$ to construct the code.
\end{thm}
\begin{proof}
Let $w = \lceil 1 + (k-1) \ln (2en) \rceil$. The algorithm described in Theorem \ref{thm:algo} constructs a $(k,n,d,w)$-superimposed code which is, clearly, a $(k,n,d)$-superimposed code.

    Using the  inequalities
    \begin{equation*}
       1 + (k-1) \ln (2en) \leq \lceil 1 + (k-1) \ln (2en) \rceil \leq 2 + (k-1) \ln (2en)
    \end{equation*}
we get
        \begin{equation*}
       \ln (2en) \leq \frac{w-1}{k-1}  \leq \frac{1 + (k-1) \ln (2en)}{k-1}.
    \end{equation*}
    Therefore,   by (\ref{eq:upperWstrong}) of Theorem \ref{thm:algo} we have that
    
    \begin{align*} 
         t &\leq 1 + ((k-1) \ln(2en) +1) d + \frac{1 + (k-1) \ln(2en) }{2(k-1)} + \frac{e}{\ln(2en)}\ \cdot \\  
        & \qquad \qquad \cdot (2 + (k-1) \ln(2en)) \left(2+(k-1) \ln(2en) - \frac{\ln(2en)}{2}+\frac{1}{2}  \right) (2en)^{\frac{1}{\ln(2en)}}\\ 
        &\stackrel{(i)}{\leq} 1 + d(k-1) \ln(2en) +d+ \frac{1}{2(k-1)} + \frac{\ln(2en)}{2} + \frac{e}{\ln(2en)}\ \cdot \\ 
        & \qquad \qquad \cdot (2 + (k-1) \ln(2en)) \left((k-1) \ln(2en) +\frac{3}{2}  \right) (2en)^{\frac{1}{\ln(2en)}}\\ 
        &\stackrel{(ii)}{=} 1 + d(k-1) \ln(2en) + d+\frac{1}{2(k-1)} + \frac{\ln(2en)}{2} + \frac{e^2}{\ln(2en)}\ \cdot \\  
        & \qquad \qquad \cdot \left (3 + \frac{7(k-1) \ln(2en))}{2} +  (k-1)^2 (\ln(2en))^2\right) \\ 
        &= 1 + d(k-1) \ln(2en) + d+\frac{1}{2(k-1)} + \frac{\ln(2en)}{2} + \frac{3e^2}{\ln(2en)} + \\ 
        & \qquad \qquad + e^2 (k-1)^2 \ln(2en) + \frac{7e^2(k-1)}{2} \\
        &\stackrel{(iii)}{\leq} d(k-1) \ln(2en) +  \frac{\ln(n)}{2} + e^2 (k-1)^2 \ln(2en) + \frac{7e^2(k-1)}{2} +d+ O(1)\,,
    \end{align*}
    where $(i)$ holds due to the fact that $\ln(2en) \geq 2$ for $n \geq 2$, $(ii)$ holds  since  $(2en)^{\frac{1}{\ln(2en)}} = e$, and $(iii)$ is 
    since $k \geq 2$ and $\ln(2en) \geq 2$.
\end{proof}

We notice  that a 
widely believed conjecture of Erd\H{o}s, Frankl and F\"uredi \cite{Erdos} says that for $k \geq \sqrt{n}$   one has that minimum-length
$(k,n,0)$-superimposed codes (i.e., classical superimposed codes)
have length $t$ equal to $n$.  The current best-known result has been proved in \cite{Shann} which shows  that if $k \geq 1.157 \sqrt{n}$ then the minimum
length of $(k,n,0)$-superimposed codes is equal to $n$. 
This last result clearly holds also for arbitrary $(k,n,d)$-superimposed codes.
We also recall the following result obtained in \cite{Olgica}.

\begin{rem}[\cite{Olgica}] \label{rem:NonExistence}
Every $(k, n,d)$-superimposed codes of length $t$ must satisfy
$$
   t \geq \min \left\{ n,  1 + (k-1)(d+1)\right\}	\,.
$$
This implies that if $k \geq \frac{n-1}{d+1} + 1$ then $t = n$, so we cannot construct a $(k,n,d)$-superimposed code of length $t$ that is better than the identity matrix of size $n \times n$.
\end{rem}

To properly appraise the value of   Theorem \ref{thm:algosuperimposed}, 
we recall the following result presented in \cite{Olgica} that provides a
lower bound on the minimum length of 
 any $(k,n,d)$-superimposed codes.

\begin{thm}[\cite{Olgica}]\label{thm:lb}
Given positive integers $k$ and $n$, with $2\leq k\leq \min\{1.157 \sqrt{n},$ $\frac{n-1}{d+1}+1\}$, the length $t$ of any $(k,n,d)$-superimposed code satisfies
\begin{equation*}
t \geq \Omega\left(\frac{k d}{\log (k d)}\log n + \frac{k^2}{\log k} \log n\right).
\end{equation*}
\end{thm}

Therefore, one can see that the construction method provided by our  Theorem \ref{thm:algosuperimposed}, besides being quite efficient, produces codes of almost optimal length.

In \cite{Olgica}, the authors provide the following upper bound on the length of $(k,n,d,w)$-superimposed codes.

\begin{thm}[\cite{Olgica}]\label{thm:OlgicaEtAl}
There exists a $(k,n,d,w)$-superimposed code of length $t$, provided that $t$ satisfies the inequality
\begin{equation*}
n \binom{n-1}{k-1} \left( \frac{w(k-1)}{t - (2d+1)(w-1)} \right)^w < 1\,.
\end{equation*}
\end{thm}

From Theorem \ref{thm:OlgicaEtAl} one can derive an explicit upper bound on the length of the codes  whose existence was showed in  \cite{Olgica}  when $w = k \ln (n)$ by upper bounding  $n \binom{n-1}{k-1}$ with  $k\left(\frac{en}{k}\right)^k$. We report here the obtained result.

\begin{thm}[\cite{Olgica}]\label{th:OlgicaThm}
There exists a randomized algorithm to construct a $(k,n,d)$-superimposed code with length
\begin{equation}\label{eq:Oult}
	t \leq 2 d k \ln (n) + k \ln(n) +  e^2 k (k-1) \ln (n) - 2d + O(1).
\end{equation}
\end{thm}

It can be seen that for $n$ sufficiently large, the upper bound on the length $t$ 
of $(k,n,d)$-superimposed codes
given in 
our Theorem \ref{thm:algosuperimposed} improves on the upper bound  
given in Theorem \ref{thm:OlgicaEtAl} of \cite{Olgica}. 
As observed in Section \ref{sec:Intro}, the algorithm provided in \cite{Olgica} is a Montecarlo randomized algorithm that constructs a $(k,n,d)$-superimposed code whose length is upper bounded by 
(\ref{eq:Oult}). In order to transform the algorithm given in \cite{Olgica}
into a Las Vegas randomized 
algorithm that \emph{always} outputs a {correct} $(k,n,d)$-superimposed code,
one can perform the following steps:
1) Generate a random matrix in accordance with the probabilities specified in 
Theorem 2 of
\cite{Olgica}, 
2) check whether the matrix  satisfies the properties of 
Definition \ref{def:dsuper} and, if not, repeat the experiment till one obtains a matrix with the desired property. 
However, it is known that the problem of checking whether a matrix satisfies superimposed-like properties is considered computationally infeasible (see, e.g., \cite{Cheng,DuK}) and no algorithm of complexity less
that $\Theta(n^k)$ is known.
On the other hand, our result   provides a randomized algorithm
of average time complexity $\Theta( k(k+d)n^2\ln n)$, that is, polynomial both in $n$ and $k$,
to construct $(k,n,d)$-superimposed codes of length not greater than that of 
\cite{Olgica}.

\end{document}